\documentclass[12pt]{amsart}
\usepackage{amsmath,amsthm,amsfonts,amscd,amssymb}
\usepackage{times,euler,eucal,eufrak,graphicx}
\usepackage{latexsym}
\usepackage{graphics}
\usepackage{hyperref}

\DeclareMathOperator{\Tr}{Tr}
\DeclareMathOperator{\tr}{tr}

\def\Wg{\mathrm{Wg}}

\newtheorem{prop}{Proposition}[section]

\newtheorem{coro}{Corollary}[section]

\theoremstyle{plain}
\newtheorem{defn}{Definition}[section]
\newtheorem{theo}{Theorem}[section]
\newtheorem{lem}{Lemma}[section]
\newtheorem{remark}{Remark}[section]
\newtheorem{Example}{Example}[section]

\title{ Compound Wishart Matrices and Noisy Covariance Matrices: Risk Underestimation}
\author{Beno\^it Collins, David McDonald, and Nadia Saad}

\hypersetup{colorlinks=true
, linkcolor=blue}


\begin{document}

\maketitle

\begin{abstract}

In this paper, we obtain a property of the expectation
of the inverse of compound Wishart matrices which results from
their orthogonal invariance. Using this property as well as results from
random matrix theory (RMT), we derive the asymptotic effect of the noise induced by estimating
the covariance matrix on computing the risk of the optimal portfolio.
This in turn enables us to get an asymptotically unbiased estimator of the
risk of the optimal portfolio not only for the case of independent observations
but also in the case of correlated observations. This improvement
provides a new approach to estimate the risk of a portfolio based on covariance
matrices estimated from exponentially weighted moving averages
of stock returns.
\end{abstract}

\section{Introduction}

In practical situations in portfolio management, neither the expectation of the returns of the assets, nor the covariance matrix is known and we always
deal with estimators. Since the estimators depend only on a finite number of observations, estimating the parameters of the portfolio produces a
noise and as the size of the portfolio increases, the noise increases.

Here, we will focus on the noise induced by estimating the covariance matrix. Covariance matrices have a fundamental role in the theory of portfolio
optimization as well as in the risk management. The concept of financial risk attempts to quantify the uncertainty of the outcome
of an investment and hence the magnitude of possible losses.

In finance, the optimal portfolio is defined as that portfolio which provides the minimum risk for a certain level of return. A portfolio's risk is the
possibility that an investment portfolio may not achieve its objectives.
Markowitz defined the risk of the optimal portfolio as the standard deviation of the return on the portfolio of assets. So, to determine the weights and the risk
of the optimal portfolio, we essentially need to estimate the covariance matrix of the returns.
To estimate the covariance matrix for a set of $n$ different assets, we need to determine $n(n+1)/2$ entries from $n$ time series of length $T$.

Throughout the paper, $n$ denotes the number of the assets of the portfolio and $T$ denotes the number of observations. If $T$ is not very large
compared to $n$, which is the common situation in the real life, one should expect that the determination of the covariances is noisy.
In \cite{laloux}, Laloux et al. show that the empirical correlation matrices deduced from financial return series contain a high amount of noise. Except
for a few large eigenvalues and the corresponding eigenvectors, the structure of the empirical correlation matrices can be
regarded as random. Unfortunately, this implies large error in estimating the risk of the optimal portfolio. Hence, Laloux et al. \cite{laloux}
conclude that ``Markowitz's portfolio optimization scheme based on a purely historical determination of the covariance matrix is inadequate". Improving
the estimation of the risk of the optimal portfolio was an essential aim for many scientists (see \cite{laloux}, \cite{plerou}, \cite{pafka-kondor},
\cite{elkaroui}).

In (\cite{laloux}, \cite{plerou}), it is found that the risk level of an optimized portfolio could be improved if prior to optimization, one
gets rid of the lower part of the eigenvalue spectrum of the covariance matrix which coincides with the eigenvalue spectrum of the ``noisy" random
matrix.

For the model of normal returns, Pafka et al. \cite{pafka-kondor} and El Karoui \cite{elkaroui}, were able to compute the asymptotic effect of the noise induced from using the maximum likelihood estimator (MLE) of the covariance matrix on estimating the risk of the optimal portfolio as the ratio $\frac{1}{\sqrt{1-n/T}}$.

In our work, we deal with a more general estimator $\widehat{\Sigma}$ of the covariance matrix $\Sigma$ that encompasses and generalizes the MLE covariance.
Our aim
is to measure the effect of the noise induced by estimating the covariance matrix not only for the independent observations but also for the correlated
observations. We essentially rely on the techniques of random matrix theory (RMT) to quantify the asymptotic effect of the noise resulting from
estimating the covariance matrix on predicting the risk of the optimal portfolio. Using this asymptotic result, simulations show that we are able to
provide an unbiased estimator of the risk of the optimal portfolio. In the case of independent observations, our result agrees with that of Pafka et
al. \cite{pafka-kondor}.

The paper is divided into four parts. In Section \ref{sec:preliminary}, we explain the optimal portfolio problem and give the notation used throughout
the paper. We also represent some tools and techniques of random matrix theory which will be essential to prove our results. Section \ref{sec:main
results} contains our main results with proofs. In Section \ref{sec:applications}, we give some applications and simulations of our result. As an applications, we obtain the impact of the noise induced by estimating the covariance matrix for the exponentially weighted moving
average (EWMA) covariance matrix. In Section \ref{sec:real data}, there will be simulations relying on real data.

\section{Preliminary}\label{sec:preliminary}

\subsection{Modern Portfolio Theory (MPT)}

MPT is a theory in finance which attempts to maximize the portfolio expected return for a given amount of portfolio risk, or alternatively to minimize
risk for a given level of expected return, by choosing the relevant weights of various assets. MPT is also considered as a mathematical form of the
concept of investment diversification with the purpose of selecting a portfolio of investment assets that has collectively
a lower risk than any individual asset.

MPT models an asset's \emph{return} as a normally distributed random variable, defines the \emph{risk} as the standard deviation of return and
models a portfolio as a weighted combination of assets so that the \emph{return} of a portfolio is the weighted combination of the assets' returns. For
a portfolio $P$ of $n$ assets, the portfolio's expected return $\mu_P$ is defined as:
\begin{equation*}
\mu_P=\sum_{i=1}^n \omega_i \mu_i
\end{equation*}
where $\omega_i (i=1,2,\dots,n)$ \ is the \emph{amount of capital invested} in the asset $i$, and
       $\{\mu_i\}$ are the \emph{mean returns} of the individual assets.

Markowitz quantified the concept of risk using the well-known statistical measures of variance and covariance as shown in \cite{markowitz}.
So, the risk $\sigma_{P}$ on the portfolio can be associated to the total variance
\begin{equation*}
\sigma_{P}^2= \sum_{i,j=1}^n \omega_i \sigma_{ij} \omega_j
\end{equation*}
where $\Sigma=(\sigma_{ij})_{i,j=1}^n$ is the covariance matrix of the returns.

The goal of portfolio optimization is to find a combination of assets $\{\omega_i\}$ that minimizes the risk of the portfolio for any given level of
expected return or, in other words, a combination of assets that maximizes the expected return of the portfolio for any given level of risk. One way to formulate this
optimization problem mathematically is the following quadratic program:

\begin{equation}\label{mean-variance optimization problem}
\left\{
\begin{array}{l}
  \min \mathbf{w}^{t} \Sigma \mathbf{w}\\
  \mathbf{w}^{t}\mathbf{\mu}=\alpha, \quad \mathbf{w}^{t}\mathbf{e}=1 \\ \end{array} \right.
\end{equation}
where $\mathbf{w}^{t}=(\omega_1, \omega_2, \dots, \omega_n)$ is the transpose of the $n-$dimensional vector $\mathbf{w}$ of the optimal weights, $\mu$ is the $n-$dimensional vector whose $i-$th entry is $\mu_i$, $\alpha$ denotes the required expected reward and $\mathbf{e}$ is the $n-$dimensional vector with $1$ in each entry.

In practice, $\Sigma$ and $\mu$ are unknown and we deal with estimators denoted by
$\widehat{\Sigma}$ and $\widehat{\mu}$, respectively. Since, in our study, we focus on the noise induced by estimating the covariance matrix and its
effect on measuring the risk. We will consider the following simplified version of the portfolio optimization problem in which we deal with risky assets; i.e. none of the assets has zero variance and the covariance matrix is non-singular:

\begin{equation} \label{simple optimization problem}
\left\{
\begin{array}{l}
  \min \mathbf{w}^{t} \Sigma \mathbf{w} \\
  \mathbf{w}^{t}\mathbf{e}=1. \\ \end{array} \right.
\end{equation}

Using the method of Lagrange multipliers, the weights of the optimal portfolio are given by:
\begin{equation} \label{weights}
\omega_i=\frac{\sum\limits_{j=1}^n \sigma_{ij}^{(-1)}}{\sum\limits_{j,k=1}^n \sigma_{jk}^{(-1)}} \quad \quad (i=1, \dots, n).
\end{equation}
where, $\Sigma^{-1}=(\sigma_{ij}^{(-1)})_{i,j=1}^n$ is the inverse of the covariance matrix $\Sigma$.

Using \eqref{weights}, the risk $\sigma_P$ can be expressed in terms of the entries of $\Sigma^{-1}$ as follows:
\begin{equation} \label{optrisk}
\sigma_P =  \frac{1}{\sqrt{ \sum\limits_{i,j=1}^n \sigma_{ij}^{(-1)}}}.
\end{equation}

\subsection{Definition of the Problem}

Since we deal with an estimator of the covariance matrix instead of $\Sigma$ itself,
then
for a portfolio with $n$ assets and time series of financial
observations of the returns of length $T$, we can define two kinds of risks;
one using $\Sigma$
that we call the
\emph{True} risk, where
\begin{equation}
\text{True risk}=\sqrt{ \mathbf{w}^{t}\Sigma \mathbf{w}},
\end{equation}
with $\mathbf{w}$ denoting the vector of the optimal weights determined by using the entries of $\Sigma^{-1}$. The other kind of risk depends on
$\widehat{\Sigma}$ and is called the \emph{Predicted} risk, where
\begin{equation}
\text{Predicted risk}=\sqrt{\hat{\mathbf{w}}^t \widehat{\Sigma} \hat{\mathbf{w}}},
\end{equation}
with $\hat{\mathbf{w}}$ denoting the vector of the optimal weights determined by using the entries of $\widehat{\Sigma}^{-1}$.
\begin{remark}
Note that, in practice, only the \emph{Predicted} risk can be computed while the \emph{True} risk is unknown.
\end{remark}
Let
\begin{equation}\label{def of Q}
Q=\frac{(True \ risk)^2}{(Predicted \ risk)^2}
\end{equation}
Our goal is to have the ratio $Q$ in \eqref{def of Q} as close as possible to one. By \eqref{optrisk}, we can write
\begin{equation} \label{Q}
Q=\frac{\sum\limits_{i,j=1}^n \hat{\sigma}^{(-1)}_{ij}}{\sum\limits_{i,j=1}^n \sigma_{ij}^{(-1)}}
\end{equation}
Clearly, this ratio is close to one as the sample size $T$ tends to infinity while $n$ remains fixed. Using results from random matrices, we are going
to consider cases where $T$ and $n$ tend to infinity and $T>n+3$. We aim to derive a deterministic bias factor which can be used to correct the above
predicted risk.

\subsection{Notation}

For an $n \times n$ matrix $M$, we denote by $\Tr(M)$ the trace of the matrix $M$ and by $\tr(M)$ the normalized trace of the matrix i.e.,
$$\tr(M)=\frac{1}{n}\Tr(M).$$

For a positive integer $k$, $[2k]=\{1,2,\dots,2k\}$. Let $S_{2k}$ be the symmetric group acting on the set $[2k]$. For $\sigma \in S_{2k}$, we attach an
undirected graph $\Gamma(\sigma)$ with vertices
$1,2,\dots, 2k$ and edge set consisting of
$$
\big\{ \{2i-1,2i\} \ | \ i =1,2,\dots,k \big\} \cup
\big\{ \{\sigma(2i-1),\sigma(2i)\} \ | \ i =1,2,\dots,k \big\}.
$$
\begin{Example}\label{example1}
Let
$\sigma= \left(\begin{array}{cccccccc} 1 & 2 & 3 & 4 & 5 & 6 & 7 & 8 \\
2 & 5 & 4 &3 &1 & 8 & 7 & 6 \end{array}\right) \in S_8$. Then the associated graph $\Gamma(\sigma)$ is illustrated as shown in Figure $(1)$.
\end{Example}
Note that we distinguish every edge $\{2i-1,2i\}$ from $\{\sigma(2i-1),\sigma(2i)\}$
even if these pairs coincide.
Then each vertex of the graph lies on exactly two edges, and
the number of vertices in each connected component is even.
If the numbers of vertices are
$2\eta_1 \ge 2\eta_2 \ge \dots \ge 2\eta_l$ in the connected components of the graph,
then we will refer to the sequence $\eta=(\eta_1,\eta_2,\dots,\eta_l)$ as the coset-type of $\sigma$, see \cite[VII.2]{Macdonald}
for details.
Denote by $\kappa(\sigma)$ the length of the coset-type of $\sigma$, or equivalently
the number of the connected components of $\Gamma(\sigma)$.

Let $M_{2k}$ be the set of all pair partitions of the set $[2k]$.
A pair partition $\pi \in M_{2k}$ can be uniquely expressed in the form
$$
\pi= \big\{ \{ \pi(1), \pi(2)\},\{ \pi(3),\pi(4)\}, \dots,
\{\pi(2k-1), \pi(2k)\} \big\}
$$
with $1=\pi(1)<\pi(3)<\cdots< \pi(2k-1)$ and $\pi(2i-1)<\pi(2i)$ $(1 \le i \le k)$.
Then $\pi$ can be regarded as a permutation

$$
\left( \begin{array}{cccc} 1 & 2 & \dots & 2k \\
\pi(1) & \pi(2) & \dots & \pi(2k) \end{array}\right) \in S_{2k}.
$$

We thus embed $M_{2k}$ into $S_{2k}$.
In particular, the coset-type and the value of $\kappa$ for $\pi \in M_{2k}$ are defined.
\begin{figure} \label{fig1}
\centering
\includegraphics[trim = 40mm 130mm 40mm 110mm,height=22.8mm,width=65mm]{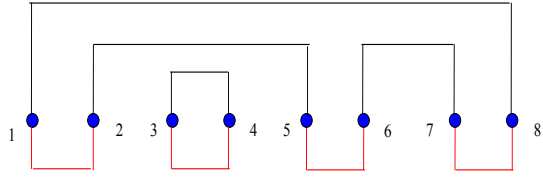}
\caption{\scriptsize{$\Gamma(\sigma)$}}
\end{figure}
For a permutation $\sigma \in S_{2k}$ and
a $2k$-tuple $\mathbf{i}=(i_1, i_2,\dots, i_{2k})$ of positive integers,
we define
$$
\delta_{\sigma}(\mathbf{i})= \prod_{s=1}^{k} \delta_{i_{\sigma(2s-1)} i_{\sigma(2s)}}.
$$
In particular, if $\sigma \in M_{2k}$, then $\delta_{\sigma}(\mathbf{i})= \prod_{\{a,b\} \in \sigma}
\delta_{i_a i_b}$,
where the product runs over all pairs in $\sigma$.
For a square matrix $A$ and $\sigma \in S_{2k}$ with coset-type $(\eta_1, \eta_2, \dots, \eta_l)$, we define
$$
\Tr_{\sigma}(A)=\prod_{j=1}^l \Tr (A^{\eta_j})
$$
\textbf{Example (continued).}
For $\sigma \in S_{8}$ defined as in Example \ref{example1}, the coset-type of $\sigma$ is $(3,1)$. This implies that
$\kappa(\sigma)=2$ and
$\Tr_\sigma(A)=\Tr (A^3)\Tr(A)$ for a square matrix $A$.

In the remaining part of this section, we present some important tools and results from random matrix theory which will play a fundamental role in
our work.

\subsection{Orthogonal Weingarten function}

For the convenience of the reader, we supply a quick review of orthogonal Weingarten calculus. For more details, we refer to
\cite{CS06,collins-matsumoto, matsumoto}.

Let $H_k$ be the hyperoctahedral group of order $2^k k!$. It is the subgroup of $S_{2k}$ generated by transpositions $(2s-1 \rightarrow 2s) (1 \leq
s \leq k)$ and by double transpositions $(2i-1 \rightarrow 2j-1)(2i\rightarrow 2j) (1 \leq i < j \leq k)$.

Let $L(S_{2k})$ be the algebra of complex-valued functions on $S_{2k}$ with the convolution. Let $L(S_{2k},H_k)$ be the subspace of all
$H_k$-biinvariant functions in $L(S_{2k})$ i.e.,
$$
L(S_{2k},H_k)= \{f \in L(S_{2k}) \ | \ f(\zeta \sigma)=f(\sigma \zeta)=f(\sigma), (\sigma \in S_{2k}, \ \zeta \in H_k)\}.
$$

We introduce another product on $L(S_{2k},H_k)$ which will be convenient in the present context.
For $f_1,f_2 \in L(S_{2k},H_k)$, we define
$$
(f_1 \sharp f_2)(\sigma)= \sum_{\tau \in M_{2k}} f_1(\sigma \tau) f_2(\tau^{-1})
\qquad (\sigma \in S_{2k}).
$$

We note that $L(S_{2k},H_k)$ is a commutative algebra under the product $\sharp$ with the
identity element
$$
\mathbf{1}_{H_k}(\sigma)= \begin{cases}1 & \text{if $\sigma \in H_k$} \\
0 & \text{otherwise}.
\end{cases}
$$

Consider the  function $z^{\kappa(\cdot)}$ with a complex parameter $z$ defined by
$$
S_{2k} \ni \sigma \mapsto z^{\kappa(\sigma)} \in \mathbb{C},
$$
which belongs to $L(S_{2k},H_k)$. For $\sigma \in S_{2k}$,
the \emph{orthogonal Weingarten function}
$\Wg^{\mathrm{O}} (\sigma;z)$
is the unique element in $L(S_{2k},H_k)$ satisfying
$$
z^{\kappa(\cdot)} \sharp \Wg^{\mathrm{O}}(\cdot;z) \sharp z^{\kappa(\cdot)} = z^{\kappa(\cdot)}
\quad \text{and} \quad
\Wg^{\mathrm{O}}(\cdot ;z) \sharp z^{\kappa(\cdot)} \sharp \Wg^{\mathrm{O}}(\cdot ;z)=
\Wg^{\mathrm{O}}(\cdot ;z).
$$

Let $\mathrm{O}(n)$ be the real orthogonal group of degree $n$,
equipped with its Haar probability measure. In \cite{collins-matsumoto}, Collins and Matsumoto formulated the local moments of the Haar orthogonal
random matrices in terms of the orthogonal Weingarten functions as shown in the following proposition.

\begin{prop}\cite{collins-matsumoto} \label{prop:OrthogonalWeingartenCalculus}
Let $O=(o_{ij})_{1 \le i,j \le n}$ be an $n \times n$ Haar-distributed orthogonal matrix.
For two sequences $\mathbf{i}=(i_1,\dots,i_{2k})$ and $\mathbf{j}=(j_1,\dots,j_{2k})$, we have
\begin{equation}
E[o_{i_1 j_1} o_{i_2 j_2} \cdots o_{i_{2k} j_{2k}}] = \sum_{\sigma,\tau \in M_{2k}}
\delta_\sigma(\mathbf{i}) \delta_{\tau}(\mathbf{j}) \Wg^{\mathrm{O}}(\sigma^{-1} \tau;n).
\end{equation}
\end{prop}

In \cite{collins-matsumoto-Saad}, the function $\Wg^{\mathrm{O}}(\cdot;z,w)\in L(S_{2k},H_k)$ was introduced
with two complex parameters $z,w$ as follows
\begin{equation}  \label{eq:double-orthogonalWg}
\Wg^{\mathrm{O}}(\cdot ;z,w) = \Wg^{\mathrm{O}}(\cdot;z) \sharp \Wg^{\mathrm{O}}(\cdot;w).
\end{equation}
The function $\Wg^{\mathrm{O}}(\cdot;z,w)$ is called double orthogonal Weingarten function. In \cite{collins-matsumoto-Saad}, Collins et al. computed
the local moments of the inverted compound Wishart matrices in terms of $\Wg^{\mathrm{O}}(\cdot ;z,w)$. We will recall this result in the next section.

\subsection{Compound Wishart matrices and their inverses}

Compound Wishart matrices were introduced by Speicher \cite{speicher}. They can be considered as an extension of the Wishart matrices. More precisely,
they are weighted sums of independent Wishart matrices.

Let $X$ be a $T\times n$ matrix of i.i.d. entries which are normally distributed with zero mean and unit variance i.e.,
\begin{equation}\label{X}
X=(x_{ij}) \quad  (i=1,\dots, T; \ j=1,\dots, n) \quad \text{such that} \quad x_{ij}\sim N(0,1).
\end{equation}
\begin{defn}(Real Compound Wishart matrices)
Let $\Sigma$ be an $n\times n$ positive definite matrix and $B$ be a $T\times T$ real matrix. We
say that a random matrix $W$ is a real compound Wishart matrix with shape parameter $B$ and scale parameter $\Sigma$, if
\begin{equation*}
W=\Sigma^{\frac{1}{2}}X^{t}BX \Sigma^{\frac{1}{2}}
\end{equation*}
where $\Sigma^{\frac{1}{2}}$ is the symmetric root of $\Sigma$.
We write $W \in \mathcal{W}(\Sigma,B)$.
\end{defn}

As shown in \cite{burda-jarosz-jurkiewicz-nowak-papp-zahed}, if $B$ is positive-definite, then $W$ can be interpreted as a sample covariance matrix.
In the sequel, we will consider $B$ to be positive definite. More details about the compound Wishart matrices can be found in \cite{hiai}. For a
compound Wishart matrix $W$, if $\Sigma=I_n$ we will call $W$ a \emph{white compound Wishart} matrix.

\begin{defn}(orthogonal invariance)
Let $M$ be an $n\times n$ real random matrix. $M$ is called orthogonally invariant if for each orthogonal matrix $O$, $OMO^t$ has the same distribution
as $M$. We write
$OMO^t\stackrel{\mathcal{L}}{=}M$.
\end{defn}
Note that white compound Wishart matrices are example of such orthogonal invariant matrices.

For a portfolio with $n$ assets and time series of financial observations of the returns of length $T$, we define a general estimator
$\widehat{\Sigma}$ of the covariance matrix $\Sigma$ as follows:
\begin{equation} \label{Estimator}
\widehat{\Sigma}=\frac{1}{\Tr(B)}Y^{t}BY
\end{equation}
where $Y=(y_{ij})$ is a $T \times n$ matrix whose rows are $n-$dimensional vectors of centered returns which are taken sequentially in time: $Y_1,
Y_2, \dots, Y_T$. We assume these vectors are i.i.d. with distribution $N(0,\Sigma)$ so that $y_{ij}$ is the return of the $j$-th asset at time $i$,
hence $Y \sim N(\mathbf{0}, I_T \otimes \Sigma)$ where $\otimes$ denotes the Kronecker product of matrices and $B$ is a $T \times T$ known weighting
matrix.
\begin{remark}
If $B=I_T$, the $T \times T$ identity matrix, then $\widehat{\Sigma}$ is the maximum likelihood estimator (MLE) of the covariance matrix which is
distributed as a real Wishart matrix with $T$ degrees of freedom. The real Wishart matrices were introduced by Wishart \cite{wishart} and as evidenced
by the vast literature, the Wishart law is of primary importance to statistics, see (\cite{anderson}, \cite{muirhead}).
\end{remark}
Since
\begin{equation}\label{Y}
Y\stackrel{\mathcal{L}}{=}X \Sigma^{\frac{1}{2}}.
\end{equation}
Then
\begin{equation}\label{sigma_hat_dist}
\widehat{\Sigma}\stackrel{\mathcal{L}}{=}\frac{1}{\Tr(B)}\Sigma^{\frac{1}{2}} X^tBX \Sigma^{\frac{1}{2}}.
\end{equation}
From \eqref{sigma_hat_dist},
$\widehat{\Sigma}$ is a compound Wishart matrix with scale parameter $\Sigma$ and shape parameter $B$.
In \cite{collins-matsumoto-Saad}, Collins et al. formulated the local moments of the inverted compound Wishart matrices as shown in the following
theorem:

\begin{theo}\cite{collins-matsumoto-Saad} \label{thm:real-compound-inverse}
Let $W$ be an $n \times n$ compound Wishart matrix $W \in \mathcal{W}(\Sigma,B)$ for some $T \times T$ real matrix $B$.
Let $W^{-1}=(w^{(-1)}_{ij})$ be the inverse matrix
of $W$. Put $q=T-n-1$ and suppose $n \ge k$ and $q \ge 2k-1$.
For any sequence $\mathbf{i}=(i_1,\dots,i_{2k})$, we have
$$
 E[ w^{(-1)}_{i_1 i_2} \cdots w^{(-1)}_{i_{2k-1} i_{2k}} ]
= (-1)^k \sum_{\sigma, \rho \in M_{2k}}
 \Tr_{\sigma} (B^{-}) \Wg^{\mathrm{O}}(\sigma^{-1} \rho;T,-q)
\prod_{\{u,v \} \in \rho} \sigma^{(-1)}_{i_u i_v},
$$
where $B^-$ is the pseudo inverse of the matrix $B$.
\end{theo}

\section{Main Results} \label{sec:main results}

In order to derive a deterministic bias factor which can be used to improve the predicted risk of the optimal portfolio, we need to derive the following
property of the inverted compound Wishart matrices.
%In which, %???
For this purpose,
we show that for a compound Wishart matrix $W$ with a scale parameter $\Sigma$ and a
shape parameter $B$, the ratio between the expected trace of $W^{-1}$ and the expected sum of its entries equals to the ratio between the trace of
$\Sigma^{-1}$ and the sum of its entries.
\begin{prop} \label{theo1}
For an $n \times n$ matrix $W \in \mathcal{W}(\Sigma,B)$,
\begin{equation*}
E(\Tr(W^{-1}))/E(\sum_{i,j=1}^n w_{ij}^{(-1)})=\Tr(\Sigma^{-1})/\sum_{i,j=1}^n \sigma_{ij}^{(-1)}.
\end{equation*}
\end{prop}

Before we prove this proposition we need to recall the following well-known fact:

\begin{lem} \label{orthoinvariance}
If $M$ is an $n\times n$ orthogonally invariant matrix, then\\
$(i)\quad E(M)=\beta I_n$, where $\beta$ is some scalar.\\
$(ii)\quad M^{k}$ is orthogonally invariant for each integer $k \in \mathbb{Z}$.
\end{lem}

\begin{proof}[Proof of Proposition \ref{theo1}]
Consider
\begin{equation} \label{A}
A=X^{t}BX.
\end{equation}
It is clear that $A$ is orthogonally invariant. By Lemma \ref{orthoinvariance}(ii) taking $k=-1$, $A^{-1}$ is orthogonally invariant as well and
\begin{equation} \label{orthogonal}
E(A^{-1})=\beta I_{n},
\end{equation}
for some scalar $\beta$. Another important remark is that,
\begin{eqnarray}
E(\sum_{i,j=1}^n w_{ij}^{(-1)}) &=& E(\Tr(\mathbf{e}^{t}W^{-1}\mathbf{e})) \nonumber\\
&=&\Tr(E(\mathbf{e}^{t}W^{-1}\mathbf{e})).
\end{eqnarray}
Since $W \in \mathcal{W}(\Sigma,B)$ then, $W^{-1}\stackrel{\mathcal{L}}{=}\Sigma^{-\frac{1}{2}}A^{-1} \Sigma^{-\frac{1}{2}}$ and so,
\begin{equation*}
E(\sum_{i,j=1}^n w_{ij}^{(-1)})= E(\Tr(\mathbf{e}^{t}\Sigma^{-\frac{1}{2}}A^{-1} \Sigma^{-\frac{1}{2}}\mathbf{e})).
\end{equation*}
Since $\Tr$ is invariant under cyclic permutations,
\begin{eqnarray*}
E(\sum_{i,j=1}^n w_{ij}^{(-1)})&=&E(\Tr(\Sigma^{-\frac{1}{2}}\mathbf{e}\mathbf{e}^{t}\Sigma^{-\frac{1}{2}}A^{-1})) \\
&=& \Tr(\Sigma^{-\frac{1}{2}}\mathbf{e}\mathbf{e^{t}}\Sigma^{-\frac{1}{2}}E(A^{-1})).
\end{eqnarray*}
It follows that
\begin{eqnarray*}
\Tr(\Sigma^{-1}) E(\sum_{i,j=1}^n w_{ij}^{(-1)})&=& \Tr(\Sigma^{-1})\Tr(\Sigma^{-\frac{1}{2}}\mathbf{e}\mathbf{e^{t}}\Sigma^{-\frac{1}{2}}E(A^{-1}))\\
&=&\Tr(\beta\Sigma^{-1}) \Tr(\mathbf{e^{t}}\Sigma^{-1}\mathbf{e}) \ \ \ (\text{from \eqref{orthogonal}})\\
&=&\Tr(E(A^{-1})\Sigma^{-1}) \sum_{i,j=1}^n \sigma_{ij}^{(-1)}\\
&=&\Tr(E(A^{-1}\Sigma^{-1})) \sum_{i,j=1}^n \sigma_{ij}^{(-1)}\\
&=&E(\Tr(W^{-1}))\sum_{i,j=1}^n\sigma_{ij}^{(-1)}.
\end{eqnarray*}
\end{proof}

\begin{remark}
Note that the $T \times T$ matrix $B$ depends essentially on the dimension $T$. So, from now on we will replace $B$ by $B_T$.
\end{remark}

In the following theorem, we study the asymptotic behavior of the ratio $Q$, defined in \eqref{def of Q},
which will play a great role in improving the prediction of the risk of the
optimal portfolio.

\begin{theo} \label{thm:noise}
Let $B_T$ be a $T \times T$ real matrix such that
\begin{equation}\label{B-condition}
(\tr(B_T))^2\tr(B_T^{-2})=o(T).
\end{equation}
Let $\widehat{\Sigma}$ be as defined in \eqref{Estimator}. If $T>n+3$, then as $n$ and $T$ tend to infinity such that $n/T\rightarrow r<1$, we have
\begin{equation} \label{effect}
Q- \Tr(B_T) E(\tr((X^tB_TX)^{-1})) \stackrel{\mathbf{\mathrm{P}}}{\longrightarrow}0.
\end{equation}
\end{theo}

\begin{remark}
The condition $T>n+3$ is related to Theorem \ref{thm:real-compound-inverse} in order to compute the second moment of the inverse of a compound Wishart
matrix and to obtain a formula for the variance of the difference $Q- \Tr(B_T) E(\tr((X^tB_TX)^{-1}))$.
\end{remark}

In order to prove Theorem \ref{thm:noise}, we need first to consider the following result concerning the variance of the ratio $Q$.

\begin{prop}\label{prop:var(Q)}
Let $B_T$ be a $T \times T$ real matrix and let $\widehat{\Sigma}$ be as defined in \eqref{Estimator}. If $q=T-n-1$, then for $q>2$,
\begin{equation*}\label{Var(Q)interms of q}
Var(Q)=\frac{(\Tr(B_T))^2}{T^2(T+2)(T-1)q^2(q-2)(q+1)}\left( a_1 (\Tr(B_T^{-1}))^2 + a_2 \Tr(B_T^{-2}) \right)
\end{equation*}
where
$$a_1=2T^2q-2Tq^2+2T^2+2T+2q^2-2q-4$$
and
$$a_2=Tq(2T-2q+2Tq-2).$$
\end{prop}

\begin{proof}[Proof of Proposition \ref{prop:var(Q)}]
\begin{eqnarray*}
Var(Q)&=&\frac{1}{(\sum\limits_{i,j=1}^n\sigma_{ij}^{(-1)})^{2}}\left(
E((\sum\limits_{i,j=1}^n\hat{\sigma}_{ij}^{(-1)})^2)-(E(\sum\limits_{i,j=1}^n\hat{\sigma}_{ij}^{(-1)}))^2 \right)\\
&=&\frac{1}{(\sum\limits_{i,j=1}^n\sigma_{ij}^{(-1)})^{2}}\left( \sum\limits_{i_1,i_2, i_3, i_4=1}^n E(
\hat{\sigma}^{(-1)}_{i_1i_2}\hat{\sigma}^{(-1)}_{i_3i_4})-\left(\sum\limits_{i,j=1}^n E(\hat{\sigma}_{ij}^{(-1)})\right)^2 \right)\\
\end{eqnarray*}
Substitute from \eqref{Estimator} to get
\begin{equation}\label{Var(Q)}
Var(Q)=\frac{(\Tr(B_T))^2}{(\sum\limits_{i,j=1}^n\sigma_{ij}^{(-1)})^{2}}\left( E(\sum\limits_{i_1,i_2, i_3, i_4=1}^n
w^{(-1)}_{i_1i_2}w^{(-1)}_{i_3i_4})-(E(\sum\limits_{i,j=1}^n w_{ij}^{(-1)}))^2 \right)
\end{equation}
where $W=(w_{ij})$ is an $n \times n$ compound Wishart matrix with scale parameter $\Sigma$ and shape parameter $B$.
By applying Theorem \ref{thm:real-compound-inverse}, we get
\begin{eqnarray*}
E(w_{ij}^{(-1)})&=& (-1) \sum_{\sigma, \rho \in M_{2}}
 \Tr_{\sigma} (B_T^{-1}) \Wg^{\mathrm{O}}(\sigma^{-1} \rho;T,-q)
\prod_{\{u,v \} \in \rho} \sigma^{(-1)}_{i_u i_v}\\
&=& (-1) \Tr(B_T^{-1}) \Wg^{\mathrm{O}}(\{1,2\};T,-q) \sigma^{(-1)}_{ij}
\end{eqnarray*}
where $q=T-n-1\geq 1$.
By using the values of $\Wg$ in \cite{collins-matsumoto}, we get
\begin{equation}\label{part1}
E(w_{ij}^{(-1)})= \frac{1}{Tq}\Tr(B_T^{-1}) \sigma^{(-1)}_{ij}.
\end{equation}
By applying Theorem \ref{thm:real-compound-inverse} again, then for $q\geq 3$ we get
\begin{multline} \label{second_moment}
E(w^{(-1)}_{i_1i_2}w^{(-1)}_{i_3i_4})=\sum_{\rho \in M_{4}}\big((\Tr(B_T^{-1}))^2 \Wg^{\mathrm{O}}(\rho;T,-q)+ \Tr(B_T^{-2})
[\Wg^{\mathrm{O}}(\pi_1\rho;T,-q)+\\ \Wg^{\mathrm{O}}(\pi^{-1}_2 \rho;T,-q)]\big) \prod_{\{u,v \} \in \rho} \sigma^{(-1)}_{i_u i_v}
\end{multline}
where $\pi_1=\{\{1,3\},\{2,4\}\}$ and $\pi_2=\{\{1,4\},\{2,3\}\}$.\\
From direct computations using \eqref{eq:double-orthogonalWg} and the values of $\Wg$ in \cite{collins-matsumoto}, we obtain the following equations:
\begin{multline}\label{Wg1}
\sum_{\rho \in M_{4}} \Wg^{\mathrm{O}}(\rho;T,-q)\prod_{\{u,v \} \in \rho} \sigma^{(-1)}_{i_u i_v}= \frac{1}{T(T+2)(T-1)q(-q+2)(q+1)}\\
\left(((T+1)(-q+1)+2)\sigma^{(-1)}_{i_1i_2}\sigma^{(-1)}_{i_3i_4}+(q-T-1)\sigma^{(-1)}_{i_1i_3}\sigma^{(-1)}_{i_2i_4}+(q-T-1)\sigma^{(-1)}_{i_1
i_4}\sigma^{(-1)}_{i_2 i_3}\right),
\end{multline}
\begin{multline}\label{Wg2}
\sum_{\rho \in M_{4}} \Wg^{\mathrm{O}}(\pi_1 \rho;T,-q)\prod_{\{u,v \} \in \rho} \sigma^{(-1)}_{i_u i_v}= \frac{1}{T(T+2)(T-1)q(-q+2)(q+1)}\\
\left((T+1)(q-T-1)\sigma^{(-1)}_{i_1i_2}\sigma^{(-1)}_{i_3i_4}+((T+1)(-q+1)+2)\sigma^{(-1)}_{i_1i_3}\sigma^{(-1)}_{i_2i_4}+(q-T-1)\sigma^{(-1)}_{i_1
i_4}\sigma^{(-1)}_{i_2 i_3}\right),
\end{multline}
and,
\begin{multline}\label{Wg3}
\sum_{\rho \in M_{4}} \Wg^{\mathrm{O}}(\pi^{-1}_2 \rho;T,-q)\prod_{\{u,v \} \in \rho} \sigma^{(-1)}_{i_u i_v}= \frac{1}{T(T+2)(T-1)q(-q+2)(q+1)}\\
\left((q-T-1)\sigma^{(-1)}_{i_1i_2}\sigma^{(-1)}_{i_3i_4}+(q-T-1)\sigma^{(-1)}_{i_1i_3}\sigma^{(-1)}_{i_2i_4}+((T+1)(-q+1)+2)\sigma^{(-1)}_{i_1
i_4}\sigma^{(-1)}_{i_2 i_3}\right).
\end{multline}
Substitute from \eqref{Wg1}, \eqref{Wg2}, and \eqref{Wg3} into \eqref{second_moment} to obtain
\begin{equation}\label{part2}
E(w^{(-1)}_{i_1i_2}w^{(-1)}_{i_3i_4})=\frac{1}{T(T+2)(T-1)q(q-2)(q+1)}\left((\Tr(B_T^{-1}))^2 I_1 + \Tr(B_T^{-2})I_2 \right),
\end{equation}
where $q>2$ and
$$I_1=\left(((T+1)(q-1)-2)\sigma^{(-1)}_{i_1i_2}\sigma^{(-1)}_{i_3i_4}+(T-q+1)\sigma^{(-1)}_{i_1i_3}\sigma^{(-1)}_{i_2i_4}+(T-q+1)\sigma^{(-1)}_{i_1
i_4}\sigma^{(-1)}_{i_2 i_3}\right),$$
and
$$I_2=\left(2(T-q+1)\sigma^{(-1)}_{i_1i_2}\sigma^{(-1)}_{i_3i_4}+(Tq-2)\sigma^{(-1)}_{i_1i_3}\sigma^{(-1)}_{i_2i_4}+(Tq-2)\sigma^{(-1)}_{i_1
i_4}\sigma^{(-1)}_{i_2 i_3}\right).$$
By substituting from \eqref{part1} and \eqref{part2} into \eqref{Var(Q)}, the proof is complete.
\end{proof}

If $B=I_T$, then $\widehat{\Sigma}$ in \eqref{Estimator} is the MLE of the covariance matrix $\Sigma$. For this case, Proposition \ref{prop:var(Q)}
reduces to the following interesting corollary.

\begin{coro}\label{cor:identityB}
Let $\widehat{\Sigma}$ be as defined in \eqref{Estimator}. If $B_T=I_T$, then for $q>2$
\begin{equation}\label{var(Q)for_identityB}
Var(Q)=\frac{2T^2}{q^2(q-2)}.
\end{equation}
\end{coro}

\begin{remark}
Note that if $B=I_T$, then Corollary \ref{cor:identityB} implies that as $n,T\rightarrow \infty$ such that $\frac{n}{T}\rightarrow r \ (r<1)$, $Var(Q)\rightarrow 0$.
\end{remark}

Now, we are ready to prove Theorem \ref{thm:noise}.

\begin{proof}[Proof of Theorem \ref{thm:noise}]
Let
$$Z_{n,T}=Q-E(\tr((\frac{1}{\Tr(B_T)}X^tB_TX)^{-1})).$$
The proof is divided into two parts. First, we show that $$E(Q)=\Tr(B_T)E(\tr((X^tB_TX)^{-1})),$$
then we will prove that for $T>n+3$
$$Var(Z_{n,T})\rightarrow 0 \quad
\text{as} \quad n,T \rightarrow \infty \quad \text{such that} \quad n/T \rightarrow r<1.$$
For the first part, apply Proposition \ref{theo1} to (\ref{Q}) to get:
\begin{equation}\label{E(Q)}
E(Q)=\frac{E(\Tr(\widehat{\Sigma}^{-1}))}{\Tr(\Sigma^{-1})}.
\end{equation}
From (\ref{Estimator}), we have
\begin{eqnarray*}
E(Q)&=&\frac{\Tr(B_T) E(\Tr(\Sigma^{-\frac{1}{2}} (X^t B_T X)^{-1} \Sigma^{-\frac{1}{2}}))}{\Tr(\Sigma^{-1})}\\
&=&\frac{\Tr(B_T)\Tr(\Sigma^{-1} E((X^t B_T X)^{-1}))}{\Tr(\Sigma^{-1})}.
\end{eqnarray*}
Since $X^tB_TX$ is orthogonally invariant then by Lemma \ref{orthoinvariance}, we obtain
\begin{equation*}
E(Q)=\frac{\beta \Tr(B_T) \Tr(\Sigma^{-1})}{\Tr(\Sigma^{-1})},
\end{equation*}
where $\beta=E(tr((X^tB_TX)^{-1}))$ which prove that $E(Z_{n,T})=0$.
This concludes the first part of the proof.\\
To complete the proof of the theorem, it is enough to show that for $T>n+3$ and as $T, n\rightarrow \infty$ such that $\frac{n}{T}\rightarrow r
(r<1)$,\ $Var(Q) \rightarrow 0$.
Suppose that
\begin{equation}\label{B-condition1}
(\tr(B_T))^2\tr(B_T^{-2})o(T).
\end{equation}
By Jensen's inequality,
\begin{equation}\label{jensen-inequality}
(\tr(B_T^{-1}))^2\leq \tr(B_T^{-2}).
\end{equation}
From \eqref{B-condition1} and \eqref{jensen-inequality}, we get
\begin{equation}\label{B-condition2}
(\tr(B_T))^2(\tr(B_T^{-1}))^2=o(T).
\end{equation}
Since $q=T-n-1$ then using Proposition \ref{prop:var(Q)} and for $T-n>3$, $Var(Q)$ can be written in terms of $n$ and $T$ as follows:
\begin{equation}\label{Var(Q)interms of n}
Var(Q)=\frac{(\Tr(B_T))^2\cdot(a^*_1 (\Tr(B_T^{-1}))^2 + a^*_2 \Tr(B_T^{-2}))}{T^2(T+2)(T-1)\cdot S(T,n)}
\end{equation}
where
$$a^*_1=2(T^2n-Tn^2+3T^2+n^2-4Tn-3T+3n),$$
$$a^*_2=2T^4-4T^3n+2T^2n^2-2Tn^2+6T^2n-4T^3+2T^2-2Tn.$$
and,
\begin{eqnarray*}
S(T,n)&=&T^4-4T^3n+6T^2n^2-4Tn^3+n^4-5T^3+15T^2n-15Tn^2+5n^3\\
&+& 7T^2-14Tn+7n^2-3T+3n.
\end{eqnarray*}
Let $R=\frac{n}{T}$ then, from \eqref{Var(Q)interms of n} and for $T>n+3$,
\begin{multline}\label{Var(Q)interms of R}
Var(Q)=\frac{(\tr(B_T))^2\big( a^{**}_1 (\tr(B_T^{-1}))^2 + a^{**}_2 \tr(B_T^{-2}) \big)}{T^3(T+2)(T-1)((1-R^3)T^3-5(1-R)^2T^2+7(1-R)T-3)}
\end{multline}
where
$$a^{**}_1=2T^5(R T^2-(3-R)T-3)$$
and
$$a^{**}_2=2T^5(T^2(1-R)-(2-R)T+1).$$
From \eqref{B-condition1}, \eqref{B-condition2} and \eqref{Var(Q)interms of R},
$$Var(Q)\rightarrow 0 \quad
\text{as} \quad n,T \rightarrow \infty \quad \text{such that} \quad n/T \rightarrow r<1.$$
\end{proof}

\begin{remark}
For $T>n+3$\ and for the case $r<1$, the convergence \eqref{effect} in Theorem \ref{thm:noise} holds if and only if $B$ satisfies the condition in \eqref{B-condition}.
\end{remark}

\begin{remark}
For $T>n+3$ and for the case $r=1$, \eqref{Var(Q)interms of R} implies that as $n, T$ tend to infinity such that $n/T\rightarrow 1$, Theorem
\ref{thm:noise} still holds if $\frac{1}{T}(\tr(B_T))^2\tr(B_T^{-2})$ converges to $0$ faster than the convergence of $1-(\frac{n}{T})^3$ to $0$.
Under this condition, our simulation shows that this result still works  for $T>n$.
\end{remark}

\begin{remark}
For the case $T \leq n$, we need to compute the moments of the inverted Wishart matrices in this case which is beyond the scope of this paper.
\end{remark}

\begin{remark}
Let $B_T$ be a diagonal matrix whose main diagonal entries $b_{ii}=e^{-i}$ for $i=1,\dots,T$. For this $B_T$, the condition \eqref{B-condition}q is not
satisfied and our simulation shows that Theorem \ref{thm:noise} is not valid too. On the other hand side, if $B_T$ is the diagonal matrix such that
$b_{ii}=i$\ for\ $(i=1,\dots,T)$ then, condition (\ref{B-condition}) is satisfied and Theorem \ref{thm:noise} holds too.
\end{remark}

By Theorem \ref{thm:noise}, to know the asymptotic value of $Q$ we need to study the asymptotic behavior of the term $\Tr(B)\tr((X^tB_TX)^{-1})$. As
shown in \cite{mardya} (page $68$), the matrix $X^tB_TX$ is a weighted sum of independent Wishart matrices and the weights are the eigenvalues of the
matrix $B_T$. So, the distribution of the matrix $X^tB_TX$ depends essentially on the eigenvalues of $B_T$. By applying Theorem
\ref{thm:real-compound-inverse} to Theorem \ref{thm:noise}, we obtain the following corollary

\begin{coro}\label{asymptotic_behavior_of_Q}
Let $B_T$ be a $T \times T$ real matrix and let $\widehat{\Sigma}$ be as defined in (\ref{Estimator}). If $T>n+3$, and $\lim\limits_{T\rightarrow
\infty}\frac{1}{T}(\tr(B_T))^2\tr(B_T^{-2})=0$ then, as $T$ and $n$ tend to infinity such that $\frac{n}{T}\rightarrow r<1$ we have
\begin{equation} \label{effect}
Q-\frac{\Tr(B_T)\Tr(B_T^{-1})}{T(T-n-1)} \stackrel{\mathbf{\mathrm{P}}}{\longrightarrow}0.
\end{equation}
\end{coro}

\section{Applications}\label{sec:applications}

In the following, we consider the case of independent observations.

\subsection{The case where $B_T$ is an idempotent}
Let $B_T$ be an idempotent matrix i.e., $B_T=B^2_T$. If $B_T$ has rank $m\leq T$ then, $B_T$ has $m$ nonzero eigenvalues and each eigenvalue equals
one. In this case, condition \eqref{B-condition} holds and $X^tB_TX$ is a
white Wishart matrix with $m$ degrees of freedom. From Corollary \ref{asymptotic_behavior_of_Q}, as $T$ and $n$
tend to infinity and $\frac{n}{T}\rightarrow r<1$,
\begin{equation}\label{idempotent}
Q-\frac{(T-k)^2}{T(T-n-1)} \stackrel{\mathbf{\mathrm{P}}}{\longrightarrow} 0.
\end{equation}

\subsubsection{\textbf{Example: Maximum Likelihood Estimator (MLE)}}
$\widehat{\Sigma}$ is a maximum likelihood estimator of $\Sigma$ if $B_T=I_T$ in \eqref{Estimator}. For this estimator, $B_T$ is an idempotent of rank
$T$. From \eqref{idempotent}, we get the following result.

\begin{coro}\label{cor:MLE}
If $\widehat{\Sigma}$ is the MLE of $\Sigma$, then
as $T$ and $n$ tend to infinity such that $\frac{n}{T}\rightarrow r<1$, we have
$$Q \stackrel{\mathbf{\mathrm{P}}}{\longrightarrow} 1/(1-r).$$
\end{coro}

This result coincides with the result of Pafka and Kondor in \cite{pafka-kondor}.
\begin{figure} \label{fig2}
\centering
\includegraphics[width=120mm]{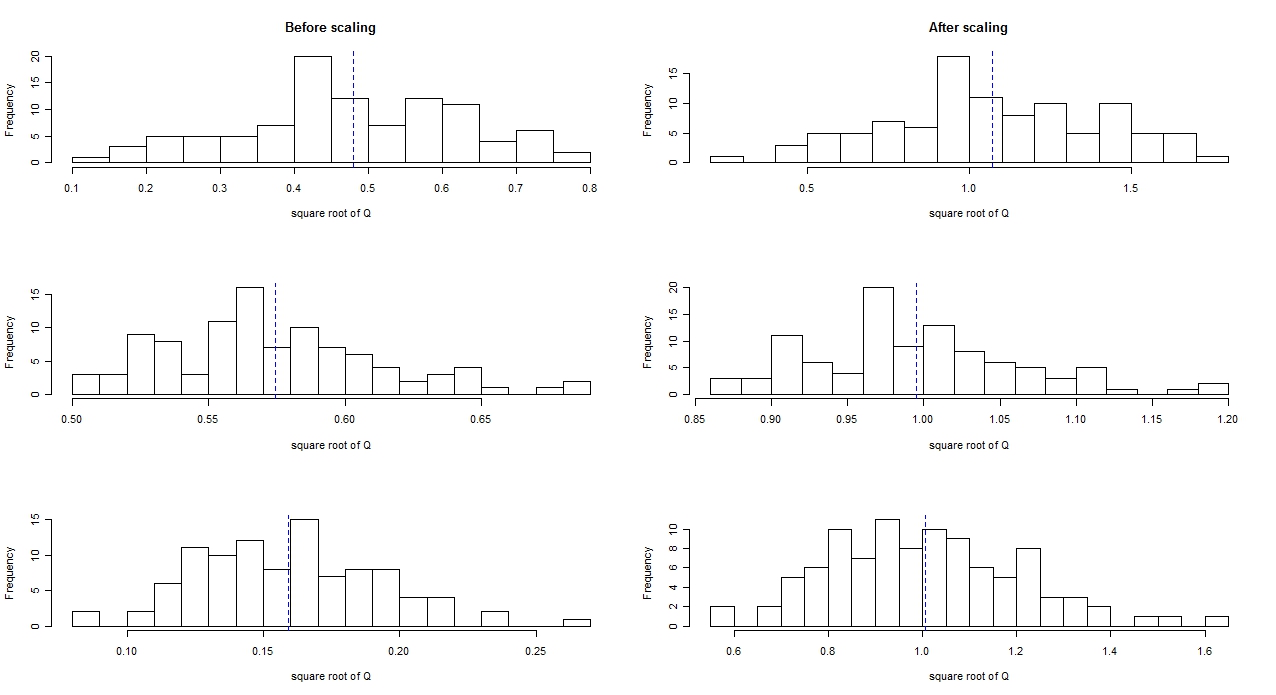}
\caption{\scriptsize{The figure illustrates the ratio between the Predicted and the True risks for the MLE before and after scaling by applying
Corollary \ref{cor:MLE}. The left side of the figure represents the ratio between the two risks before scaling while the graphs on the right hand side of the
figure describe the histogram of the ratio between the risks after applying Corollary \ref{cor:MLE}. The middle part of the
figure illustrates the ratio between the risks when $n=200$ and $T=250$. In the upper part of the figure, we focus on the case of small values of $n$
and $T$ ($n=20, T=25$) while in the lower graphs, we choose $n$ and $T$ with close values ($n=390$ and $T=400$). The mean of the ratio between the
Predicted and the True risks, represented by a dotted line in each histogram, shows a valuable improvement in estimating the Predicted risk after
scaling the Predicted risk using Corollary \ref{cor:MLE}.}}
\end{figure}

To
simulate this case, we randomly choose a
value for $\Sigma$. Using this value, we compute the True risk. Second, we generate a
set of observations from the distribution $N(0,\Sigma)$ and estimate $\mu$ and $\Sigma$ using these observations. Finally, we compute the Predicted
risk using the MLE covariance $\widehat{\Sigma}$ and compare the Predicted and the True risks. As shown in Figure ($2$), we get a remarkable improvement in estimating the risk for MLE after scaling the Predicted risk using the factor
$\frac{1}{\sqrt{1-\frac{n}{T}}}$ in Corollary \ref{cor:MLE}. The figure illustrates the ratio between the Predicted and the True risks before and after applying Corollary \ref{cor:MLE}. The dotted line in each histogram represents the mean of the ratio between the two risks. For the middle graphs of the
figure, we take $n=200$ and $T=250$ and for these values the mean of the ratio between the risks before and after scaling equals $0.575$ and $0.996$,
respectively which shows a remarkable improvement in computing the Predicted risk.

To study the validity of the Scaling technique for small values of $n,T$, we take $n=20$ and $T=25$ and as shown in the upper graphs of Figure $(2)$,
the mean of the ratio between the risks before and after scaling is $0.464$ and $1.037$, respectively. So, the Scaling technique is still valid for small dimensions and small observations situations.\\
In the lower graphs of Figure $(2)$, we choose closed values for $n$ and $T$ ($n=390$ and $T=400$) and the mean of the ratio between the risks equals
$0.159$ and $1.007$ before and after scaling, respectively. From the simulations, we conclude that for the MLE, the Scaling technique is a real
improvement in estimating the risk. Also, note that the reduction in the standard deviation of the ratio of the Predicted and the True risks from the upper graph to the middle graph as $n$ and $T$ increases from $n=20$ and $T=25$ to $n=200$ and $T=250$. In theory, the standard deviation goes to zero an $n$ and $T$ tend to infinity such that $n/T \rightarrow r \ (r<1)$ by Corollary \ref{cor:MLE}.

\subsection{When the Expected Returns are unknown} The unbiased estimator of the covariance matrix is called the sample covariance matrix and is given by
$$\widehat{\Sigma}=\frac{1}{T-1}Y^tY.$$
The sample covariance estimator can be obtained from \eqref{Estimator} by considering the entries of the matrix $B$ as follows:
$$b_{ii}=1-\frac{1}{T} \ \text{for} \ (i=1,\dots,T)  \quad \text{and} \quad b_{ij}=-\frac{1}{T} \ \text{for} \ (1<i<j<T).$$
In this case, $B$ is an idempotent of rank $T-1$. In \cite{elkaroui}, El-Karoui shows that the asymptotic behavior of the noise resulting from
estimating the covariance matrix using the sample covariance estimator (with unknown expected means of the returns) is
$\frac{1}{\sqrt{1-\frac{n-1}{T-1}}}$ which still coincides with our result in \eqref{idempotent} although in our case we assume the
returns are centered. This similarity between the two cases is due to the independence between the estimators $\hat{\mu}$ and $\widehat{\Sigma}$. To
simulate this case, we randomly choose values for $\mu$ and $\Sigma$. Using these values, we compute the True risk. Next, we generate a
set of observations from the distribution $N(\mu,\Sigma)$ and estimate $\mu$ and $\Sigma$ using these observations. Finally, we compute the Predicted
risk using the estimators $\hat{\mu}$ and $\widehat{\Sigma}$ and compare the Predicted and the True risks. As shown in Figure ($3$), the ratio between
the scaled Predicted risk and the True risk is very close to one and there is a valuable improvement in estimating the Predicted risk after using the
Scaling technique.

In the next section, we are going to study an important estimator of the covariance matrix which plays a great role in many fields, specially in finance

\begin{figure} \label{fig3}
\centering
\includegraphics[width=120mm]{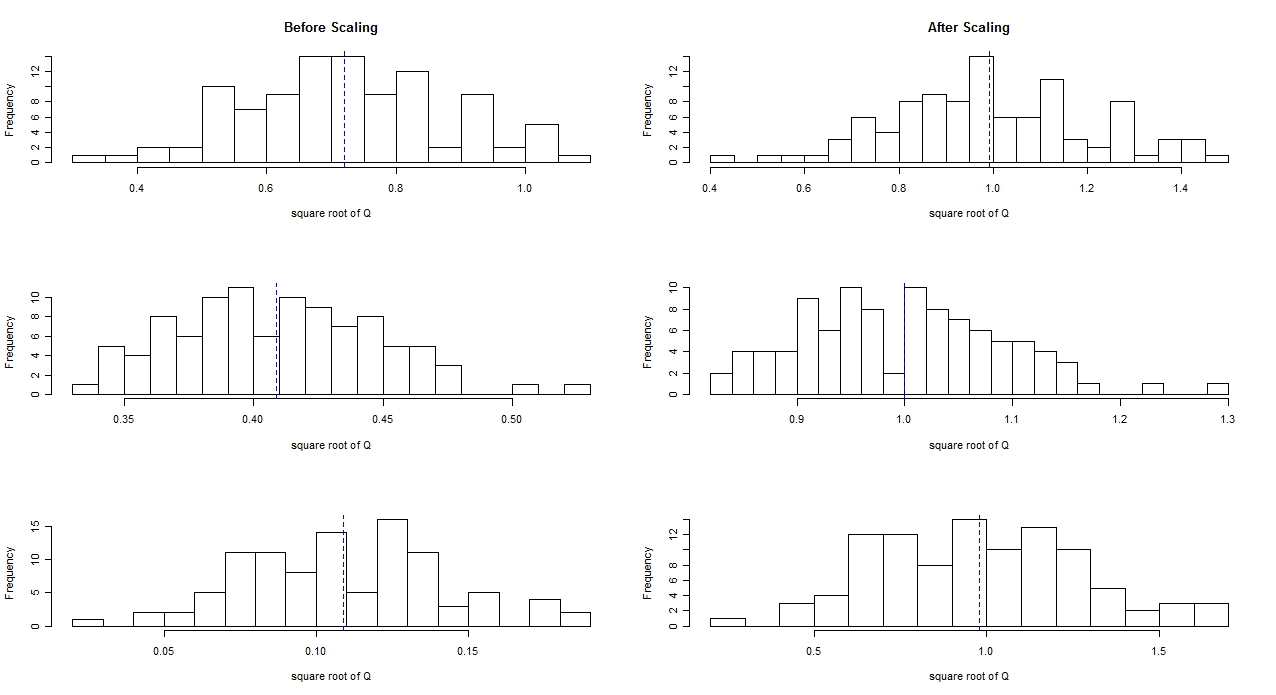}
\caption{\scriptsize{The figure describes the ratio between the \emph{Predicted} and the \emph{True} risks for the sample covariance estimator in the
case of unknown expected mean of the returns. As shown in the graphs on the right hand side, there is a real improvement in estimating the risk of the
optimal portfolio after scaling the predicted risk by the factor $\frac{1}{\sqrt{1-\frac{n-1}{T-1}}}$ according to \eqref{idempotent}. We
take $n= 10, 250, 400$ and $T= 20, 300, 405$, respectively. Comparing the graphs before scaling (on the left) and the graphs after scaling (on the
right), it is clear that the mean of the ratio between the Predicted and the True risks (represented by the dotted line in each histogram) becomes
closer to one after using the Scaling technique for small or large values of $n$ and $T$.}}
\end{figure}

\subsection{Exponentially Weighted Moving Average (EWMA)}

In the stock market, using equally weighted data doesn't accurately exhibit the current state of the market. It reflects market conditions which are perhaps no longer
valid by assigning equal weights to the most recent and the most distant observations. To express the dynamic structure of the market, it is better to
use exponentially weighted variances. Exponentially weighted data gives greater weight to the most recent observation. Thus, current market conditions are taken into consideration more accurately. The EWMA model is proposed by Bollerslev \cite{bollerslev}. Related studies (\cite{fama}, \cite{tse}) are made in the equity market and using exponentially weighted moving average techniques (weighting recent observations more heavily than older observations). In \cite{akgiray}, Akgiray shows that using EWMA techniques are more powerful than the equally weighted scheme.\\
In EWMA technique,
returns of recent observations to distant ones are weighted by multiplying each term starting from the most recent observation by an exponential decay
factor $\lambda^0, \lambda^1, \lambda^2, \dots, \lambda^j, \dots \ (0<\lambda<1$) respectively. In common, $\lambda$ is called the decay factor.
Hence, $b_{ij}=\delta_{ij}\lambda^{j-1}$ in \eqref{Estimator}, for $i,j=1 \dots T$ and we have
$$
\Tr(B_T)\Tr(B_T^{-1})= \frac{(1-\lambda^T)^2}{\lambda^{T-1}(1-\lambda)^2}.
$$
If $\lambda\rightarrow 1$, then $\lim\limits_{T\rightarrow \infty}\frac{1}{T}(\tr(B_T))^2\tr(B_T^{-2})=0$. Now, let us apply Theorem \ref{thm:noise} to
the EWMA estimator and obtain the following corollary.
\begin{coro}\label{EWMA}
Let $\widehat{\Sigma}$ be the EWMA estimator of the covariance matrix $\Sigma$ with decay factor $0<\lambda<1$. If $T>n+3$ then, as $\lambda$ tend to
$1$ and as $T,n$ tend to infinity such that $(1-\lambda)T=c$ (for some positive constant $c$) and $n/T\rightarrow r<1$, we have
$$
Q\stackrel{\mathbf{\mathrm{P}}}{\longrightarrow} (e^c-1)^2/c^2(1-r)e^c.
$$
\end{coro}
\begin{figure} \label{fig4}
\centering
\includegraphics[width=120mm]{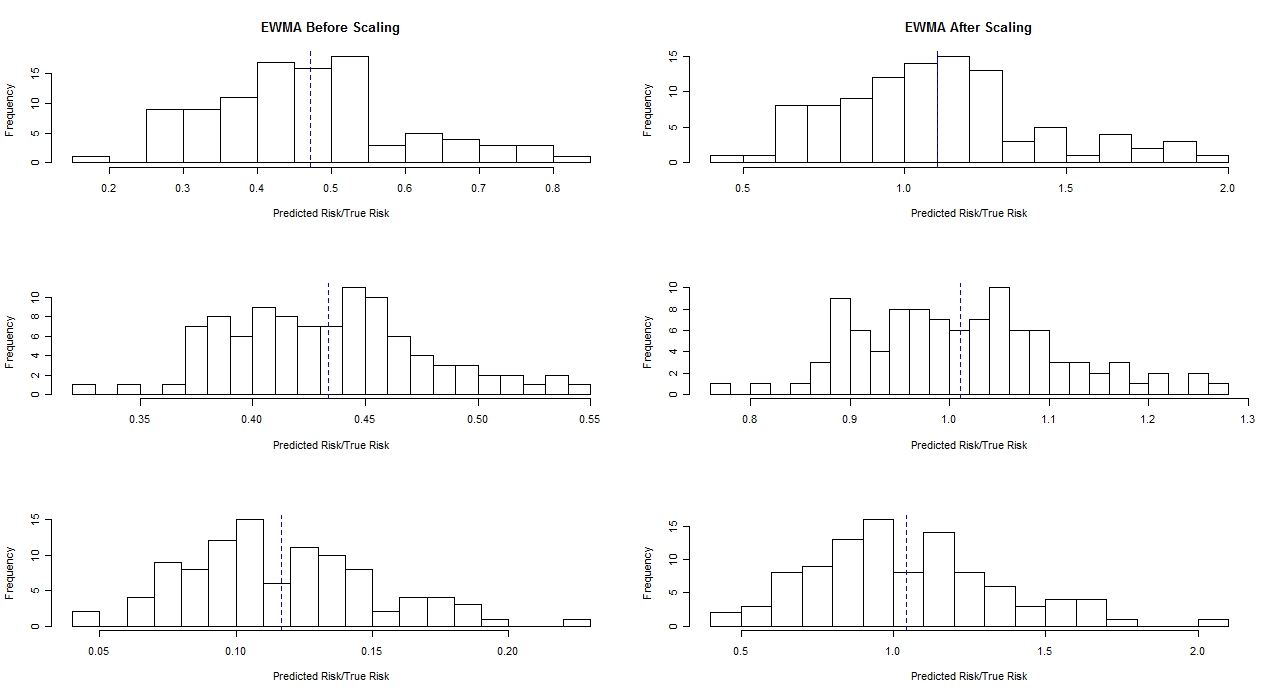}
\caption{\scriptsize{The figure describes the ratio between the \emph{Predicted} and the \emph{True} risks for the EWMA covariance estimator before and
after scaling using Corollary \ref{EWMA}. In the first row, we take small values for $n,T$, ($n=20$, $T=25$, and $\lambda=0.96$). The means of the
histograms of the upper graphs, represented by the dotted line in each histogram, equal $0.47$ (before scaling) and $1.099$ (after scaling). In the
second row, we take $n=200$, $T=250$, and $\lambda=0.996$. The means of the histograms before and after scaling are $0.43$ and $1.01$, respectively. In
the lower graphs, $n=395$, $T=400$, and $\lambda=0.9996$, and the means of the histograms equal $0.12$ (before scaling) and $1.04$ (after scaling).
Comparing the graphs before scaling (on the left) and the graphs on the right (after scaling), it is clear that the ratio between the Predicted and the
True risks becomes closer to one after using the Scaling technique.}}
\end{figure}
As shown in Figure $(4)$, for the EWMA covariance matrices, scaling the Predicted risk using Corollary \ref{EWMA} gives a great improvement to estimate
the risk of the optimal portfolio. Before scaling as illustrated in the graphs on the left hand side of Figure $(4)$, the ratio between the two risks
is far from $1$ specially for close values of $n$ and $T$ ($n=395, T=400$) as shown in the lower left graph of the figure. After scaling the Predicted
risk by the factor $\frac{exp(c)-1}{c\sqrt{(1-r)exp(c)}}$ as in Corollary \ref{EWMA}, the ratio between the Predicted and the True risks becomes very
close to $1$ as in the right hand sides graphs of the figure. For small values of $n$ and $T$, as in the upper graphs of the figure, $n=20$ and $T=25$,
the means of the histograms of the upper graphs, represented by the dotted line in each histogram, equal $0.47$ (before scaling) and $1.099$ (after
scaling). So, the Scaling technique still works and improves the estimation of the Predicted risk. Again note the reduction in the standard deviation of the ratio of the Predicted and the True risks from the upper graph to the middle graph as $n$ and $T$ increases from $n=20$ and $T=25$ to $n=200$ and $T=250$.

\section{Real Data}\label{sec:real data}

In this section, we work with real data from the stock market and observe the effect of using the scaling technique on improving the prediction of the risk of the optimal portfolio. For $30$ stocks, $n=30$, we compute the True risk using a large number ($354$) of observations. In order to compute the Predicted risk, we use only $50$ observations i.e., $T=50$.

To compute the Predicted risk, we randomly choose $50$ observations and use them to find the MLE (or EWMA) of the covariance matrix and then invert the MLE (or EWMA) and calculate the Predicted risk. After repeating this process for $100$ times, we histogram the ratio between the Predicted and the True risks before and after scaling using the result of Corollary \ref{cor:MLE} (or Corollary \ref{EWMA} in the case of EWMA covariance).

In Figure $(5)$, we illustrate the ratio between the Predicted and the True risks in the case of the MLE covariance. As shown in the upper histogram, the average of the ratio between the risks is $0.631$ before scaling. While after scaling, the average of the ratio between the risks is $0.998$ as shown in the lower histogram. This shows that using Corollary \ref{cor:MLE}, in the case of MLE covariance, admits a real improvement in estimating the risk of the optimal portfolio.

In the case of the EWMA covariance, we choose some value for the decay factor e.g. $\lambda=0.98$ and illustrate the ratio between the Predicted and the True risks before and after scaling using Corollary \ref{EWMA} as shown in Figure $(6)$. In the upper histogram, the average of the ratio between the risks is $0.64$ before scaling. While after scaling, as shown in the lower histogram, the ratio between the risks becomes $1.008$. Hence, for the EWMA covariances, using the result of Corollary \ref{EWMA} provides a better estimation of the risk of the optimal portfolio. We conclude that the scaling technique admits a good prediction of the risk of the optimal portfolio for different covariance matrices.

\begin{figure} \label{fig5}
\centering
\includegraphics[width=120mm]{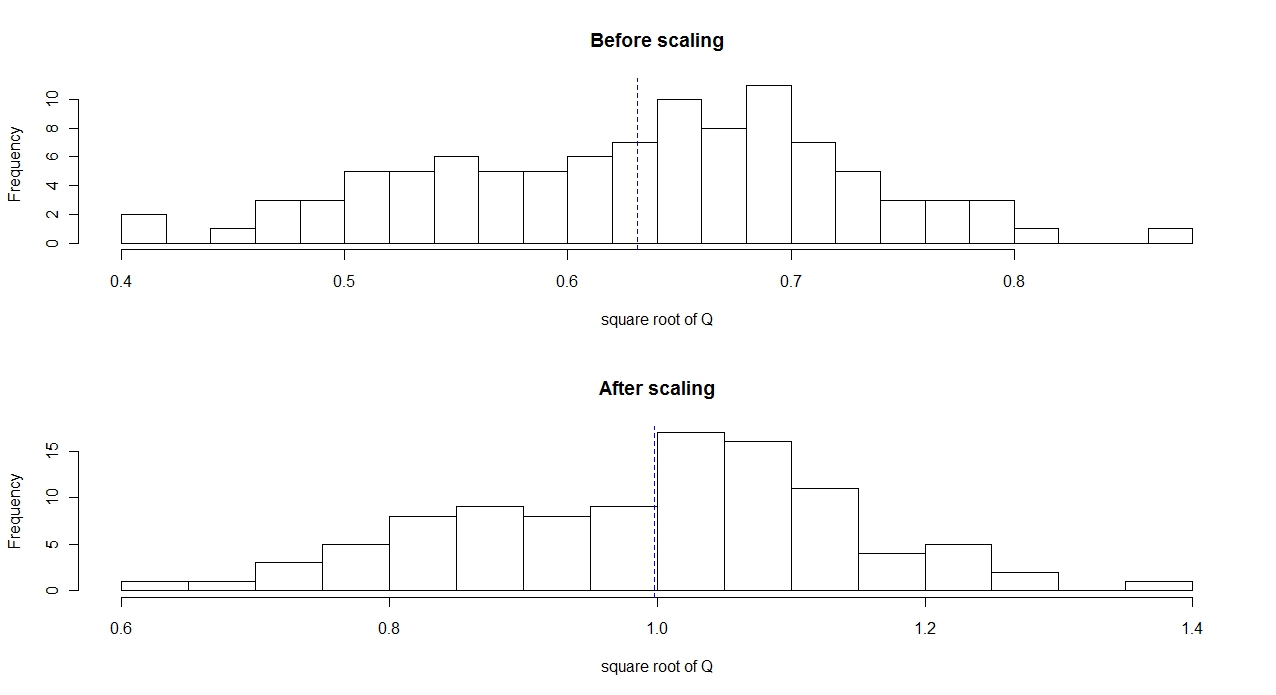}
\caption{\scriptsize{The figure describes the ratio between the \emph{Predicted} and the \emph{True} risks for the MLE covariance using real data. The upper histogram describes the ratio between the risks before scaling while the lower histogram describes the ratio after scaling using Corollary \ref{cor:MLE}. It is clear that there is a real improvement in estimating the risk of the optimal portfolio after scaling the predicted risk.}}
\end{figure}

\begin{figure} \label{fig6}
\centering
\includegraphics[width=120mm]{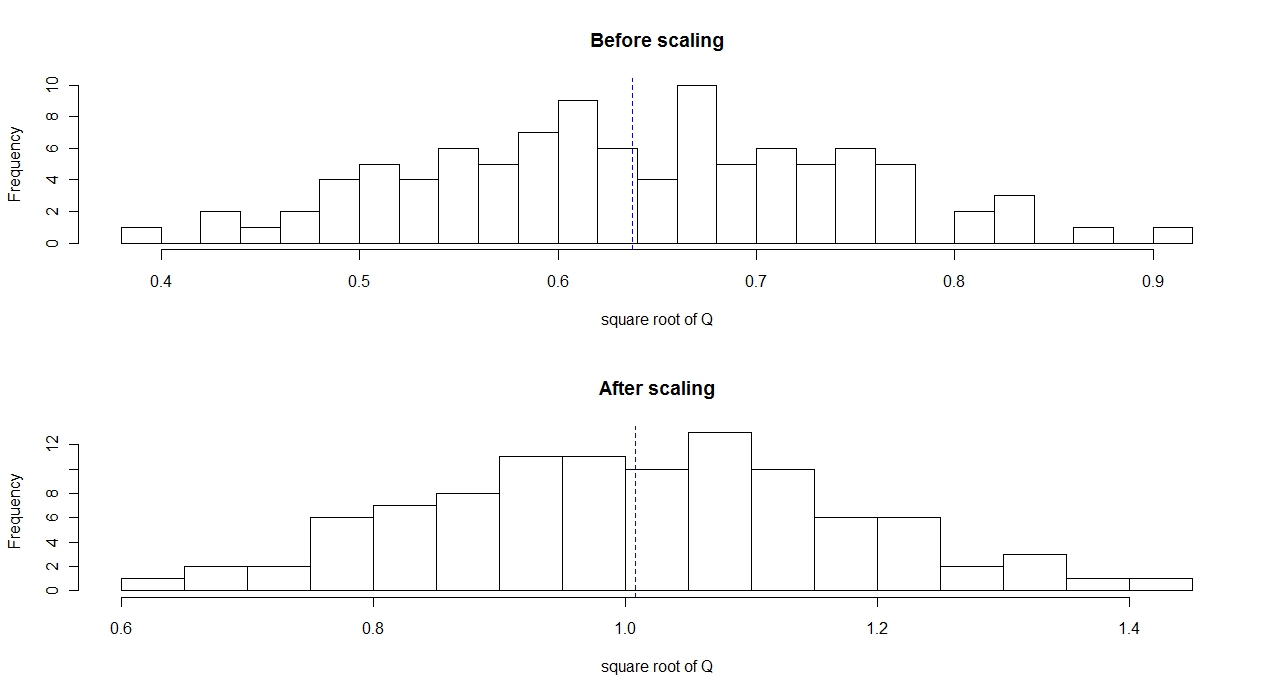}
\caption{\scriptsize{The figure describes the ratio between the \emph{Predicted} and the \emph{True} risks for the EWMA covariance using real data from the stock market. As shown in the graphs, the average of the ratio between the two risks (presented by the dotted line in each histogram) becomes closer to one after scaling in the lower graph.}}
\end{figure}

\begin{remark}
In the case of EWMA, we take different values for the decay factor $\lambda$ and in each time the ratio between the Predicted and the True risks becomes closer to one after using the Scaling technique.
\end{remark}

\section{Conclusion}\label{sec:conclusion}
For a general estimator of the covariance matrix and using our results concerning the moments of the inverse of the compound Wishart matrices, we are able to use a Scaling Technique to cancel the asymptotic effect of the noise induced by estimating the covariance matrix of the returns on the risk of an optimal
portfolio. As an application, we get a new approach on estimating the risk based on estimating the covariance matrices of stocks returns using the
exponentially weighted moving average. Simulations show a remarkable improvement in estimating the risk of the optimal portfolio using the Scaling
technique which outperforms the improvement obtained by using the Filtering technique \cite{biroli-bouchaud-potters}.

We believe that the effect of noise on computing the risk and the weights of the optimal portfolio results from estimating the inverse of the
covariance matrix (using the inverse of the estimator of the covariance matrix) not from estimating the covariance matrix itself. Improving the
estimator of the inverse of the covariance matrix is an interesting topic for our future work.

\section*{Acknowledgments}

B.C, D.McD. and N.S. were supported by NSERC discovery grants
B.C. and N.S. were supported by an Ontario's ERA grant.
B.C. was supported in part by funding from the AIMR.

The authors would like to thank M. Alvo, M. Davison, R. Kulik and S. Matsumoto for support and enlightening discussions.

\noindent
\textsc{Beno\^\i t Collins} \\
D\'epartement de Math\'ematique et Statistique, Universit\'e d'Ottawa,
585 King Edward, Ottawa, ON, K1N6N5 Canada,
WPI AIMR, Tohoku, Sendai, 980-8577 Japan
and
CNRS, Institut Camille Jordan Universit\'e  Lyon 1,
France
\verb|bcollins@uottawa.ca|

\bigskip

\noindent
\textsc{David McDonald} \\
Department of Mathematics and Statistics, University of Ottawa,
585 King Edward, Ottawa, ON, K1N6N5 Canada \\
\verb||

\bigskip

\noindent
\textsc{Nadia Saad} \\
Department of Mathematics and Statistics, University of Ottawa,
585 King Edward, Ottawa, ON, K1N6N5 Canada \\
\verb|nkotb087@uottawa.ca|

\end{document}